\documentclass[11pt, letterpaper]{article}

\usepackage[letterpaper,margin=1in]{geometry}
\usepackage[T1]{fontenc}
\usepackage[utf8]{inputenc}
\usepackage{amsmath,amssymb,amsthm}
\usepackage{enumerate}
\usepackage{color}
\usepackage{mathtools}
\usepackage{thm-restate}
\usepackage{tikz}
\usepackage{appendix}
\usepackage{tcolorbox}

\usepackage[Algorithm,ruled]{algorithm}
\usepackage[noend]{algpseudocode} 
\usepackage{cite}

\usepackage{geometry}
\definecolor{darkgreen}{rgb}{0,0.5,0}
\usepackage{hyperref}
\hypersetup{
    unicode=false,       	
    colorlinks=true,     	
    linkcolor=blue,      	
    citecolor=darkgreen, 	
    filecolor=magenta,   	
    urlcolor=cyan        	
}

\usepackage[capitalize, nameinlink]{cleveref}
\crefname{theorem}{Theorem}{Theorems}
\crefname{lemma}{Lemma}{Lemmas}
\crefname{corollary}{Corollary}{Corollaries}
\crefname{observation}{Observation}{Observations}

\newcommand{\nin}{\not\in}

\renewcommand{\epsilon}{\varepsilon}

\newcommand{\eps}{\varepsilon}

\newcommand{\calR}{\mathcal{R}}

\newcommand{\E}{\mathbb{E}}
\newcommand{\OO}{\mathcal{O}}
\newcommand{\NN}{\mathbb{N}}

\newcommand{\nwn}[1]{\texttt{newneighbor(#1)}}
\newcommand{\bnwn}[1]{\texttt{newABneighbor(#1)}}

\newcommand{\ignore}[1]{}

\newtheorem{theorem}{Theorem}[section]
\newtheorem{lemma}[theorem]{Lemma}

\newtheorem{proposition}[theorem]{Proposition}
\newtheorem{corollary}[theorem]{Corollary}
\newtheorem{definition}[theorem]{Definition}
\newtheorem{claim}{Claim}

\title{An O(n) time algorithm for finding Hamilton cycles with high probability}

\newcommand*\samethanks[1][\value{footnote}]{\footnotemark[#1]}
\author{ 
Rajko Nenadov
	\samethanks[1]
\and
Angelika Steger
	\samethanks[1]
\and
Pascal Su
	\thanks{
		Department of Computer Science, ETH Zurich, Switzerland. \newline
		Email: {\tt raikon@gmail.com, asteger@inf.ethz.ch, sup@inf.ethz.ch}.
	}\ \thanks{The research leading to these results has received funding from grant
no.\ 200021 169242 of the Swiss National Science Foundation}
}
\date{}

\begin{document}

\maketitle

\begin{abstract}	
We design a randomized algorithm that finds a Hamilton cycle in $\OO(n)$ time with high probability in a random graph $G_{n,p}$ with edge probability $p\ge C \log n / n$. This closes a gap left open in a seminal paper by Angluin and Valiant from 1979. 

\end{abstract}

\section{Introduction}
A Hamilton cycle is a cycle in a graph that visits every vertex exactly once. 
Determining whether a graph has a Hamilton cycle is a notoriously difficult problem that has been tackled in various ways. In general, it is known to be $\mathcal{NP}$-hard, putting it in a bag of complexity theory together with colorability or SAT, problems for which one has tried to find polynomial time algorithms for a long time without any success so far. 

While the Hamilton cycle problem is a difficult problem in general, it turns out that for most graphs it is actually not. To illustrate this, we take a closer look at the Erd\H{o}s-R\'enyi random graph $G_{n,p}$ which is an $n$-vertex graph with each edge being present independently with probability $p$. The existence question of the Hamilton cycle problem is very well understood, cf.\ the comprehensive survey by Frieze~\cite{frieze2019hamilton}.  
Let $\mathcal{H}$ be the set of Hamiltonian graphs, then for $G_{n,p}$ it holds that (Koml{\'o}s and Szemer{\'e}di~\cite{komlos1983limit} and Korshunov~\cite{korshunov1976solution})

\[ Pr[  G_{n,p(n)} \in \mathcal{H} ]  =  \left\{
                \begin{array}{ll}
                  0, &\quad p(n) =  \frac{\log(n) + \log\log(n) - \omega(1)}{n}  \\
                  e^{-e^{-c}},& \quad p(n) = \frac{\log(n) + \log\log(n) +c + o(1)}{n}  \\
                  1, & \quad p(n) = \frac{\log(n) + \log\log(n) + \omega(1)}{n}  \\
                \end{array}
              \right.
              \]

Which is limitwise the same as the threshold for when $G_{n,p}$ has minimum degree $2$. So really vertices of degree one are the bottleneck for random graphs. In fact, it is known that if we add the edges randomly one by one, the moment we reach minimum degree $2$ is the same as the moment the graph becomes Hamiltonian with high probability \cite{ajtai1985first}.
And this threshold is also robust (e.g. \cite{nenadov2019resilience, montgomery2019hamiltonicity}). For other random graph models like the random graph with $m$ edges $G_{n,m}$, the random regular graph $G_{n,r}$ or the $k$-out which takes $k$ random edges from every vertex the corresponding thresholds for Hamiltonicity are also known \cite{fenner1984hamiltonian, bollobas1990hamilton, bohman2009hamilton, robinson1994almost, sudakov2008local}. Similar to the classical random graph case also in these cases the thresholds coincides with a local obstruction such as minimum degree two or any two vertices have a neighborhood of size at least 3. And this is not a coincidence. Randomness gives us such nice expansion properties that only the small structures can be an obstruction to the Hamilton cycle. This phenomenon has been observed also for other properties such as connectivity, containing a perfect matching or colorability. 

 The proofs of Komlos and Szemeredi and Korshunov are just existential, i.e. they determine the threshold for the existence of Hamilton cycle, but do not provide an efficient algorithm for finding it. In a seminal paper, Angluin and Valiant  \cite{angluin1979fast} show that with the input given as a random adjacency list one can find Hamilton cycles in $G_{n,p}$ for $p\ge C \log n /n$ in $\OO( n \log^2 n) $ time with high probability. There are two ways in which this result is possibly non-optimal: the lower bound on $p$ and the runtime. The first point was considered by Shamir and then Bollobas, Fenner and Frieze, who brought the bound down to the existence threshold of $G_{n,p}$. In more recent works the runtime has also been optimized for graphs given in adjacency matrix form, assuming a pair of vertices can be queried in constant time. We summarize these results in the table below.
There are various related results that are hard to compare, as their setting is slightly different \cite{ferber2016finding, allen2015tight, frieze1996generating, frieze1988finding}. Some of the results are assuming the graph is given as an adjacency matrix with black box queries and the runtime $\OO(n/p)$ is optimal in that model.  

\begin{center}
\begin{tabular}{ l l l l l }
\hline
\bf Authors &\bf Year &\bf Time &\bf  p(n) & Graph Model\\ \hline  %
 Angluin, Valiant \cite{angluin1979fast} & `79 &$ \OO(n \log^2(n))$ & $p\ge \frac{C \log(n) }{n} $ &  adj. list \\ \hline
 Shamir \cite{shamir1983many} & `83 & $\OO(n^2)$ & $p\ge \frac{\log(n) + (3+\epsilon)\log\log(n) }{n} $ &  adj. list \\ \hline
 Bollobas, Fenner, Frieze \cite{bollobas1987algorithm} & `87 &$ n^{4 + o(1)} $ & $p\ge Existence \ threshold $ & adj. list  \\ \hline
 Gurevich, Shelah \cite{gurevich1987expected}& `87 &$\OO(n/p) $ & $p$ const.&  adj. matrix\\ \hline
 Thomason \cite{thomason1989simple} & `89 &$ \OO(n/p) $ & $p\ge C n^{-1/3} $ &  adj. matrix \\ \hline
 Alon, Krivelevich \cite{alon2020finding} & `20 &$\OO(n/p)$ & $p\ge 70n^{-1/2} $ & adj. matrix \\ \hline
\end{tabular}
\end{center}

In this paper we consider the second question that was left open in the  Angluin-Valiant paper: can the runtime be improved. Note that a graph with $p\ge C \log n /n$ has $\Theta(n\log n)$ edges. Thus, improving the runtime below
this bound requires a \emph{sublinear} algorithm, i.e.\ sublinear in the input size. These are algorithms that produce an output without reading the input completely (see e.g.\ \cite{rubinfeld2011sublinear} for an overview of the topic). Such algorithms are less restrictive than those designed for online or a (semi-)streaming model as they allow some control over which part of an input is used. However for graphs with $n$ vertices and $m \gg n$ edges the algorithm is only allowed to read $o(m)$ edges, i.e., a negligible fraction of the input --- but nevertheless has to compute the desired output correctly.

\subsection{Our contribution}
In this paper we show that given a random graph with edge probability  $p\ge C \log n /n$, for an appropriately chosen constant $C$, we can find a Hamilton cycle in $\OO(n)$ time with high probability. This time is clearly optimal, as the algorithm has to return $\Omega(n)$ edges. We assume that the graph is given to us with randomly ordered adjacency lists, such that we can query the next neighbor in those lists for any vertex in constant time.

\begin{theorem} \label{thm:main}
There exists a randomized algorithm $\calR$ which finds a Hamilton cycle in a random graph $G_{n,p}$ in $\OO(n)$ time with high probability, provided $p\ge C \log n / n$ for a sufficiently large constant $C$.
\end{theorem}

Note that `with high probability' is always meant to mean with probability $1-o(1)$ tending to one as $n$ tends to infinity and takes into account all sources of randomness: i.e., the randomness of the algorithm, the random graph and the randomness of the datastructure used to store the graph (random ordering of the adjacency lists). 

Our paper is organized as follows.
Section \ref{sec:algo} contains the algorithm and the proof of Theorem \ref{thm:main}. It is based on three technical lemmas that we prove in Section~\ref{sec:datastructures}.

\section{Algorithm} 
\label{sec:algo}

The most commonly used technique for efficient cycle extensions is Posa rotations.  
This is also the case for the original algorithm of Angluin and Valiant~\cite{angluin1979fast}, which we outline in Section~\ref{sec:oldalgo} below, cf.\ also Figure~\ref{fig:oldalgo_posa}.To reduce the runtime to $\OO(n)$ we reduce the total {\em number} of Posa rotations that are required and simultaneously also restrict ourselves to certain types of Posa rotations so that we can realize each of them in $\OO(\log n)$ time.

\subsection{Finding A Hamilton Cycle via Posa Rotations}
\label{sec:oldalgo}

We sketch here the algorithm of Angluin and Valiant. 
The main idea of their algorithm is to perform a greedy random walk until all vertices are incorporated in the path/cycle. This means we start from an arbitrary vertex and query a neighbor of that vertex. If the neighbor is already contained in the path we have built so far we consider this a failure and we query a new neighbor. Otherwise we add the neighbor to the path and continue from the new endpoint vertex (see Figure~\ref{fig:oldalgo}). 

Once the path is long enough (at least $n/2$) we add possible Posa rotations. Assume we start with a path $P = (v_1, \dots, v_s)$, then if we find two edges such that for some index $i\in [s]$ the edges are of the form  $\{v_{i+1}, v_s\}$ and $\{v_i, v_j\}$ for some $j > i+1$, we can rearrange the path to form a new path $P' = (v_1, \dots, v_i, v_j, v_{j+1}, \dots, v_s, v_{i+1}, \dots, v_{j-1})$ and now the new path has the same vertex set but a different endpoint vertex. This we call a Posa rotation. Additionally we will always want \emph{long} Posa rotations meaning $s-i$ must be at least $n/2$ to ensure that we can find the second edge needed quickly with high probability.

So during our Algorithm if the neighbor ($v_{i+1}$) of the endpoint of the path ($v_s$) has distance at least $n/2$ from the endpoint along the path we use that edge to build a cycle and continue from the vertex preceding the neighbor ($v_i$) on the path (see Figure~\ref{fig:oldalgo_posa}). This leaves a cycle of size at least $n/2$ and if we ever find one of the vertices on the cycle to be the neighbor of the current endpoint we reincorporate the large cycle by appending it to the path (again giving a new endpoint).

\begin{figure} 
\centering \includegraphics[width = \textwidth]{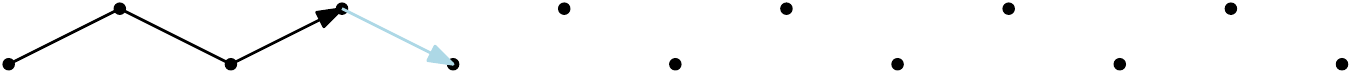} 
\caption{Algorithm uses a random walk like strategy, blue edge is the \nwn{}.}
\label{fig:oldalgo}
\vspace{1em}
\end{figure} 

\begin{figure}
\centering \includegraphics[width = \textwidth]{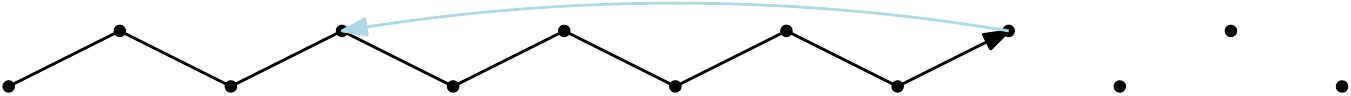} 
\centering \includegraphics[width = \textwidth]{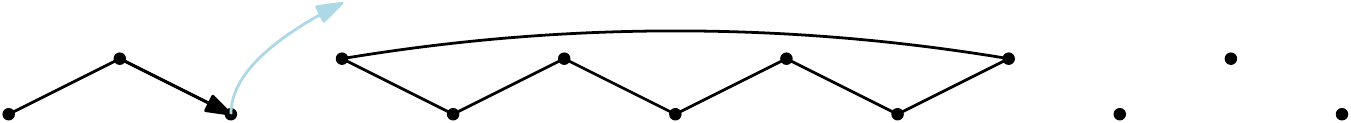} 
\caption{Posa rotation, detaching a large cycle.}
 \label{fig:oldalgo_posa}
\vspace{1em}
\centering \includegraphics[width = \textwidth]{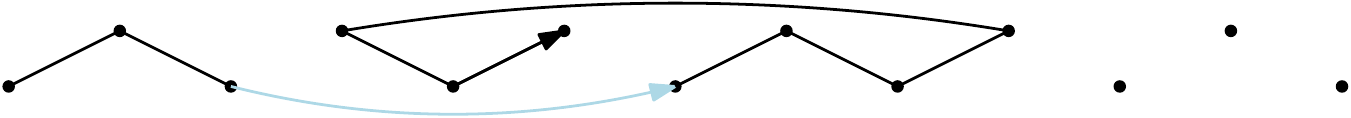}
\caption{Reincorporating the large cycle.}
\end{figure}

Many details need to be considered on how random variables interact, etc., but leaving those aside one can easily convince oneself that on average the current vertex changes after a constant number of queries to a new random vertex, and that the number of queries until the path length increases by one is geometrically distributed and has an expectation of $n/(n-i)$ where $i$ is the current length of the path. The total number of Posa rotations is thus bounded by
\[ \OO\left(\sum_{i=1}^n \frac{n}{i}\right) =\OO(n \log n) .\]
As each Posa rotation takes time $\log n$ to realize this gives a total running time of $\OO(n \log^2 n)$.

\subsection{Our Algorithm}

We give a short overview of the new algorithm we propose. The algorithm comes in two phases. In phase 1 we find two random perfect matchings. 
The union of these two random perfect matchings forms a two regular graph, i.e., a set of disjoint cycles or double edges covering all vertices. It is not difficult to show that the number of cycles is with high probability bounded by $2 \log n$. In phase 2 of the algorithm 
we stitch these $2 \log n$ cycles together.

For the analysis of the algorithm it is very helpful to assume that a query for a new neighbor of some vertex $v$ returns a vertex $w$ that is {\em uniformly} distributed over all vertices in $V - v$ and {\em independent} from all previous queries. Of course such an assumption a priori does not hold if we simply return the next vertex from the adjacency list of $v$. We realize this by directing the edges and resampling. More formally, we will show the following lemma in Section~\ref{sec:datastructures}; in the remainder of Section~\ref{sec:algo} we will use the corresponding function \nwn{} as a black box.

\begin{restatable*}[newneighbor]{lemma}{newneighbor}
\label{lem:newneighbor} It is possible to interact with the graph $G_{n, p}$, $p \ge \frac{C \log n}{n}$, with an algorithmic procedure \nwn{$v$} which has the following properties with high probability: \\
\noindent
$(i)$ Calling \nwn{$v$} returns a neighbor of $v$ distributed uniformly among $V- v$ and independent of all calls so far -- as long as we make at most $\OO(n)$ calls to \nwn{} altogether and every vertex is queried at most $100 \log n$ times.\ \
\noindent $(ii)$ The total run time of all $\OO(n)$ calls is $\OO(n)$.
\end{restatable*}
Note that this algorithm uses both internal randomness as well as the randomness of $G_{n,p}$. If \nwn{$v$} ever returns 'there are no more neighbors' we immediately terminate the entire algorithm and return failure. To avoid this, we will prove that we query \nwn{$v$} from any vertex at most $100 \log n$ times w.h.p.\ and choose $C$ large enough so that with high probability the minimum degree of the random graph is large enough.

\subsubsection{Phase 1: Perfect Matching}
\label{sec:phase1}

In the first phase of the algorithm we show that we can find a perfect matching in $\OO(n)$ time. We call the algorithmic procedure described in this section \texttt{FastPerfectMatching}, see Algorithm~\ref{alg:perfectmatching}.
%
%
In fact, for an easier understanding of the required ideas, we work in this section with a random {\em bipartite}
random graph. This can easily be done by partitioning the vertex set $V$ into two equal sets $A$ and $B$ arbitrarily (if $n$ is odd we set one vertex aside and include it in phase 2) and only considering the edges between $A$ and $B$.
Formally, the function
\bnwn{v} calls  \nwn{v} until we receive a neighbor which is in $B$ (resp. $A$).
\begin{claim}
	If we call \bnwn{} for a sequence of $\OO(n)$ vertices, in which every vertex $v\in A\cup B$ occurs at most $\log n$ times, then with high probability this results in at most $\OO(n)$ calls to \nwn{} with at most $6 \log n$ calls per vertex.
\end{claim}
The claim holds because any call of \nwn{} has probability at least $1/2$ to be in the correct partition and, by our assumptions on \nwn{}, the calls are independent. We can thus apply concentration bounds for binomial distributions and union bound for every vertex. Clearly, \bnwn{} still has a uniform and independent distribution over all vertices of the opposite partition.

Let $G$ be the balanced bipartite graph with partitions $A$ and $B$. During the algorithm we will maintain a matching $M$ which covers some of the vertices and is empty at first. At any point in time, we denote by $A_{M}$  the vertices in $A$ that are covered by the matching and with $A_{0}$ the unmatched vertices. Equivalently for $B_{M}$ and $B_{0}$. 

Additionally we need a set of edges that expand well from the vertices of $A$. And we need to be able to keep track of them efficiently and on the fly. So for any vertex $v$ we define the $d$-neighborhood of $v$, $N_d(v) \subseteq V(G)$, to be the set of the first $\lceil d\rceil$ calls to the function \bnwn{$v$}. In particular this implies that for any $d' < d$ the $d'$-neighborhood is contained in the $d$-neighborhood of $v$. Similarly, the $d$-neighborhood of a set of vertices $S$, denoted by $N_d(S)$, is defined as the union of the $d$-neighborhoods of all vertices in $S$. We expose and keep track of the $d$-neighborhood of the unmatched vertices $A_0$, $N_{d(|A_0|)}(A_0)$, for the function $d(t) = \min(\sqrt{n/t}, \log n)$. This gives us a neighborhood large enough for the random walks to be effective, but small enough so that we do not need too much time to update/expose.

To increase the matching we call a subroutine \texttt{IncreaseMatching}.
\texttt{IncreaseMatching} takes as argument the current matching $M$ and an unmatched vertex $v \in B_0$. It proceeds as follows.
If $v$ is in $N_d(A_0)$ we add the corresponding neighbor in $A_0$ and $v$ to the matching.
If not we take $w = \bnwn{v}$. If $w$ is in $A_0$ we add the edge $\{w, v\}$ to $M$. If neither of the two is the case, then $w \in A_M$ and there exists a unique $u$ such that $\{w, u\}$ is currently in $M$. We swap $\{w, u\}$ for $\{w, v\}$, thereby making $u$ a new unmatched vertex, and repeat \texttt{IncreaseMatching} with $u$, cf. Figure~\ref{fig:increasematching}.

Clearly, during the run of the algorithm we also have to dynamically update the $d$-neighborhood of $A_0$. In particular this means removing $N_d(w)$ of a newly matched vertex $w$ and, if $d(|A_0|)$ increases, adding vertices from additional calls to \bnwn{} for every vertex in $A_0$.

\begin{figure}[t]
\centering \includegraphics[width = \textwidth]{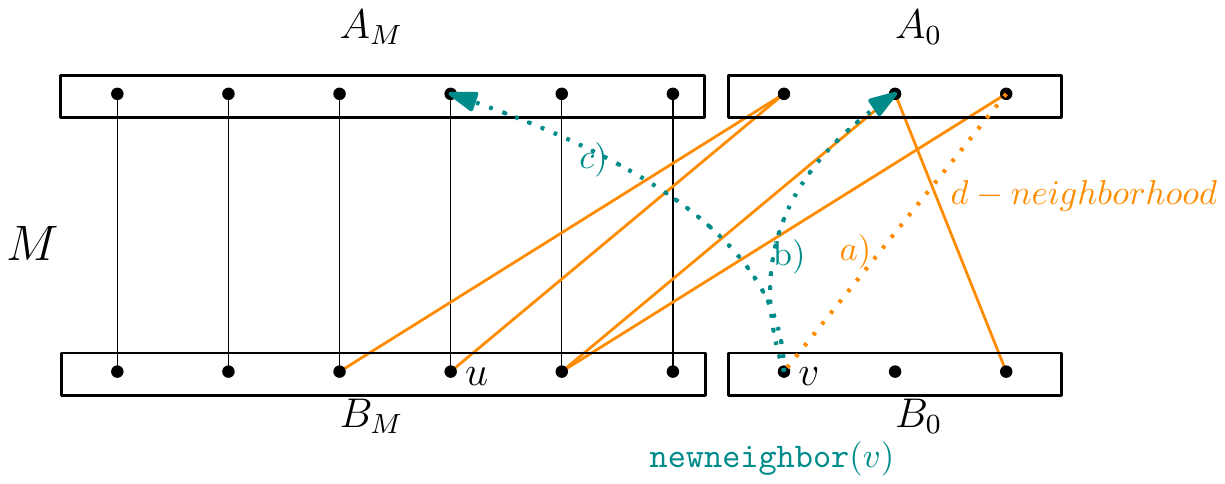} 
\caption{For \texttt{IncreaseMatching} three things can happen. Either a) the vertex is already in the neighborhood of $A_0$, in which case we match immediatly, b) the vertex \bnwn{$v$} is in $A_0$, which also gets matched, or c) \bnwn{$v$} is in $A_M$. Then we swap the matching and continue from the partner of the \bnwn{$v$}.}
\label{fig:increasematching}
\vspace{1em}
\end{figure}

\begin{algorithm}[t]
\caption{{$ FastPerfectMatching (G)$}} \label{alg:perfectmatching}
\begin{algorithmic}[1]

\State $B_0 \leftarrow B $; $B_M \leftarrow \{\} $; $A_0 \leftarrow A $; $A_M \leftarrow  \{\} $; 
\State d $\leftarrow$ 0; $M \leftarrow \{\}$
\While{$B_0 \ne \{\} $}

\State $v \leftarrow $ arbitrary vertex from $B_0$ \Comment{and remove from $B_0$}
\State IncreasingMatching($G, M, v$); \Comment{see Algorithm~\ref{alg:random_walk}}
\While{ $d < \min\left(\sqrt{\frac{n}{|A_0|}} , \log(n)\right) $ }
\State $d \leftarrow d+1$
\State Add \bnwn{$v$} to the $d$-neighborhood for every vertex in $v \in A_0$   
  \EndWhile
\EndWhile
\State \Return Matching M
\end{algorithmic}
\end{algorithm}

\begin{algorithm}[t]
\caption{{$ IncreaseMatching (G, M, v)$}}
\label{alg:random_walk}
\begin{algorithmic}[1]
\If{$v  \in N_d(A_0)$}
\State $w \leftarrow $ [neighbor of $v] \in A_0$
\State Add $\{v, w\}$ to $M$
\State Remove $w$ from $A_0$ and update $N_d(A_0)$
\State \Return
\EndIf

\State $w \leftarrow $ \bnwn{$v$}
\If{$w \in A_0$} 
\State Add $\{v, w\}$ to $M$
\State Remove $w$ from $A_0$ and update $N_d(A_0)$
\State \Return
\EndIf

\State $u \leftarrow $ unique vertex with $\{u, w\} \in M$
\State Remove $\{u, w\}$ from $M$ and replace with $\{v, w\}$
\State IncreaseMatching(G, M, u)

\State \Return

\end{algorithmic}
\end{algorithm}

To bound the runtime of Algorithm 2, \texttt{FastPerfectMatching}, we observe first that 
we increase the matching exactly $n$ times, which is inline with our desired bound of $\OO(n)$. We can thus concentrate on bounding the {\em recursive} calls to \texttt{IncreaseMatching} in line 13 of 
 \texttt{IncreaseMatching}.
 
\begin{lemma} \label{lem:randomwalk}
Let $\mathcal{L}_i$ denote the number of calls \texttt{IncreaseMatching} in line 13, while $|A_0| = i$ for any $i\in [n]$. Then the $\mathcal{L}_i$ are dominated by independent geometric distributions with success probability $p_i = \frac{i \cdot d(i)}{100 n}$.
\end{lemma}

\begin{proof}
Whenever we are at a vertex $v$ in $B$ we expose an edge to a random neighbor in the set $A$. If that vertex is in $A_0$ we match $v$ and $|A_0|$ decreases by one so we end the count of $\mathcal{L}_{|A_0|}$. Otherwise we swap with a matched vertex and get a new starting point in $B_0$. As \bnwn{} is independent and uniform, and the matching forms a bijection between $A_M$ and $B_M$, the fact that the vertex is not in $A_0$, implies that we get a new {\em random} vertex $u$ in $B_M$ for the next call. If this vertex is in the exposed $d$-neighborhood of $A_0$ we stop and match to a vertex in $A_0$ also ending the count of $\mathcal{L}_{|A_0|}$.

 To assess the probability of stopping, we use the {\em expansion properties} of the $d$-neighborhood of $A_0$ that are inherited from the random graph. This means in particular that the exposed neighborhood of $A_0$, $N_{d(|A_0|)}(A_0)$, has size at least $\frac{1}{100} |A_0| \cdot d(|A_0|) $, cf.\  Lemma~\ref{lem:neighborhood} in Section~\ref{sec:datastructures} for a proof. The probability of hitting a vertex in $A_0$ or the $d$-neighborhood of $A_0$ (while looking at the matched vertex of $w$ in $B_M$) is thus at least $\frac{ |A_0| \cdot d(|A_0|) }{100 n}$. Every new call of \bnwn{} is independent by Lemma \ref{lem:newneighbor}, thus $\mathcal{L}_i$ is dominated by an independent geometric distribution with success probability as claimed.
\end{proof}

We are now ready to proof the desired complexity bound:

\begin{proposition} \label{prop:fastalgo}
	\texttt{FastPerfectMatching} finds a perfect matching in a balanced random bipartite graph in time $\OO(n)$ with high probability. 
\end{proposition}

\begin{proof}
There are two main contributions to the running time of the Algorithm. First the subroutine \texttt{IncreaseMatching}, which we prove to be fast with the help of Lemma \ref{lem:randomwalk}, and secondly the updating and revealing of the $d$-neighborhood.

Recall that $\mathcal{L}_i$ is the random variable corresponding to the number of calls of \texttt{IncreaseMatching} in line~13, while $|A_0| = i$  for any $i\in [n]$. We set $\mathcal{L} = \sum_{i=1}^n \mathcal{L}_i$. Note that we can ignore the calls in line 5 of \texttt{FastPerfectMatching}, as these add only at total of $\OO(n)$ to the run time.
From Lemma~\ref{lem:randomwalk} we know that there exists a  coupling to a geometrically distributed random variable $\mathcal{L'}$ such that $\mathcal{L'}_i \succeq \mathcal{L}_i$ and $\mathcal{L'}_i$ \emph{is} geometrically distributed with $p_i = \frac{i \cdot d(i)}{100 n}$.

From the definition of $\mathcal{L'}_i$ we know that  $\mathbb{E}[ \mathcal{L'}_i] =   \frac{100n}{i\cdot d(i)} $ and $Var[\mathcal{L'}_i] = \frac{1-p_i}{p_i^2} \le \frac{1}{p_i^2} \le  ( \frac{100 n}{ i \cdot d(i) } )^2 $.
Recall that $d(i) = \sqrt{n/i}$ whenever $i \ge 
 \frac{n}{(\log n)^2}$. The total time used for those sets can thus be bounded in expectation by
\[\sum_{i=\frac{n}{(\log n)^2}}^n \mathbb{E}[\mathcal{L'}_i ] = \OO\left(\sum_{i=1}^n \frac{ \sqrt{n}}{  \sqrt{i} } \right) =  \OO(n), \]
as $\sum_{i=1}^n i^{-1/2} \le \int_0^n \! x^{-1/2} \, \mathrm{d}x   = 2\sqrt{n}$.
If $i \le \frac{n}{\log(n)^2}$, then $d(i)=\log n$, and the total expected time used for these sets is thus bounded by
\[\sum_{i=1}^{\frac{n}{(\log n)^2}}  \mathbb{E}[\mathcal{L'}_i] = \OO\left( \sum_{i=1}^{n}  \frac{ n}{  i \cdot \log(n) } \right)  =  \OO(n). \]
We thus have that $\E[\mathcal{L'}] = \Theta(n)$ as well. To show that the actual run time is concentrated around the expectation we apply Chebyshev's inequality. A similar case distinction as above gives us

\begin{eqnarray*}
Var[\mathcal{L'}] &\le&  
\sum_{i=\frac{n}{(\log n)^2}}^n \frac{10000 n}{  i } + \sum_{i=1}^{\frac{n}{(\log n)^2}}  \frac{10000 n^2}{  i^2\cdot  (\log n)^2 } \; = \; \OO\left(\frac{n^2}{(\log n)^2} \right).
\end{eqnarray*}
By Chebyshev's inequality we thus get

\[Pr[\mathcal{L} \ge 2\E[\mathcal{L'}]] \le Pr[\mathcal{L'} \ge 2\E[\mathcal{L']}] \le \frac{Var[\mathcal{L'}]}{(\E[\mathcal{L'])}^2} \le \OO\left( \frac{1}{(\log n)^2} \right),\]
which concludes the first part of the proof.
\vspace{1em}

To bound the time needed to expose the $d$-neighborhoods, we observe first that we can order the vertices in $A$ by the order in which they join the matching. As $N_{d'}(v) \subseteq N_d(v)\  \forall d'\le d$, we thus have to expose for the $i$-th vertex in this ordering  at most $d(n-i ) +1$  edges, where $d(x) = \min\{\sqrt{n/x},\log n\}$. Thus, the total number of exposed edges is bounded by

\begin{eqnarray*}
\sum_{i = 1}^n   (d(n-i )+1) &=&    \sum_{i = 1}^{\frac{n}{(\log n)^2}}   (d(i )+1) + \sum_{i = \frac{n}{(\log n)^2}}^n   (d(i )+1 )                                            \\
&\le & \sum_{i = 1}^{\frac{n}{(\log n)^2}}( \log n + 1) + \sum_{i = \frac{n}{(\log n)^2}}^n          \sqrt{\frac{n}{i}}\;\le\;  4n    .         
\end{eqnarray*}     

\noindent

Additionally we show the number of calls to the \bnwn{} function is at most $\log n$ for every vertex w.h.p.. For any $v \in B$ we call \bnwn{$v$} exactly once for each time it appears as the matched partner of \bnwn{$v$}. As the distribution on the  neighbors is uniform on $A$ and we only use \texttt{IncreaseMatching} $\OO(n)$ many times in total, the probability that $v \in B$ occurs at least $\log n$ times is at most

\[ \binom{ \OO(n)}{\log n} \left(\frac{1}{n}\right)^{\log n} =  \OO(n^{-2}), \]
with room to spare. We can thus apply a union bound over all vertices in $B$ to see that w.h.p.\ no vertex 
in $B$  has more than $\log n$ calls to $\bnwn{}$. Clearly the same holds for vertices in $A$, as we only expose the $d$-neighborhood and $d(|A_0|) \le \log n$ always. 
This concludes the proof of Proposition~\ref{prop:fastalgo}.
\end{proof}

\subsubsection{Phase 2: Incorporating the Cycle Factor}

In the previous section we have seen that we can find a perfect matching in $\OO(n)$ time. In this section we show how we can extend this algorithm to find a Hamilton cycle. To do this we first call the perfect matching algorithm {\em twice}, reseting the $d$-neighborhoods after the first run. By our assumption on the independence on the calls to the function \nwn{}, we thereby get two {\em independent} random perfect matchings. Their union forms a union of cycles (or double edges) covering all vertices (if the number of vertices was odd we add the single vertex excluded in phase 1 here back as a cycle with one vertex). Our task in this phase is to join these cycles into a single cycle.
We start with a lemma that bounds the number of cycles that we need to join.

\begin{lemma}
The union of two random independent perfect matchings in a bipartite graph contains at most $2\log n$ cycles with high probability.
\end{lemma}

\begin{proof}
We claim that
the two independent perfect matchings can be seen as a random permutation of $[n/2]$. Indeed, without loss of generality we may assume that $M_1$ is just the identity (by renumbering the vertices appropriately). $M_1$ and $M_2$ are independent which implies $M_2$ corresponds to a random assignment of $B$ to $A$. The union of the two matchings thus defines a random permutation of $A$.

For random permutations the number of cycles has been well studied and is related to the Stirling numbers of the first kind. Using a double counting argument one can easily see that the expected number of cycles of length $2k$ will be $1/k$. The total expected number of cycles is thus equal to the $n$th harmonic number. It is also well-known that this random variable is concentrated, see e.g. \cite{arratia2003logarithmic} or \cite{arratia1992cycle, maples2012number}. Thus,  with high probability the number of cycles is bounded by $2 \log n$, as claimed.
\end{proof}

\noindent \textbf{Description of Algorithm~\ref{alg:joincycles} \texttt{JoinCycles}.}
To glue the cycles together we proceed in three phases. First we greedily combine cycles into a path, until this path has length at least $3n/4$. Then we incorporate the remaining cycles one by one using Algorithm~\ref{alg:addsinglecycle} \texttt{AddSingleCycle}. Finally, we close the Hamilton path into a Hamilton cycle (Lemma~\ref{lem:closecycle}).

The idea behind the first phase is straightforward. We start with an arbitrary cycle and break it apart into a path $P$. Consider the endvertex $p_{end}$ of that path. We use \nwn{} to query a new neighbor of $p_{end}$. If that neighbor is in a new cycle (which will happen with probability at least $1/4$, as long as the path $P$ contains at most $3n/4$ vertices), we  attach that cycle to $P$, thereby also getting a new endpoint $p_{end}$. If the latter did not happen, we query a new neighbor.
In order to ensure that we do not query to many vertices from a single vertex, we repeat the query for new neighbors at most $40\log n$ times. If we have not been successful by then, we give up. It is easy to see that the probability for ever giving up at this stage of the algorithm is bounded by $o(1)$. It is also easy to see that the total time spent until the path has length at least $3/4n$ is bounded by $\OO(n)$.

Once the path has length at least $3n/4$, the probability that a new neighbor is in one of the remaining cycles gets too small (for our purpose) and we thus change strategy. In particular, we add long Posa rotations, so that we can try various endpoints. This is the purpose of the procedure \texttt{AddSingleCycle} (Algorithm~\ref{alg:addsinglecycle}). 

We use a set $U$ to keep track of {\em used} vertices. Those are vertices for which we already queried neighbors within the algorithm \texttt{JoinCycles}. We denote the current path by $P = (p_{start}, .. ,p_{end})$. We also assume that we have access to a function $pred_P(v)$ that determines the vertex before $v$ on the path (null for $p_{start}$), and a function $half_p(v)$ which is true iff $v$ is in the first half of $P$. We denote the cycle $C$ that we want to add as $C = (c_{start},...,c_{end})$, where  $c_{start}$ is an arbitrary vertex at which we cut $C$ into a path. We now explore neighborhoods of vertices at once. To do this we denote by \nwn{$v, 40\log n$}  the set of vertices that we obtain if we apply \nwn{$v$} $40\log n$  times. Let $N_{start} = P\; \cap\; $ \nwn{$c_{start}, 40\log n$} and  $N_{end} =  P \;\cap\; $ \nwn{$c_{end}, 40\log n$} denote the intersections of these neighborhood vertices with the path $P$. 
Until the cycle $C$ is part of the path $P$ we do the following (Figure~\ref{fig:addsinglecycle}).
Let $N(p_{end}) =$ \nwn{$p_{end}, 40\log n$} and check for all $v\in N(p_{end})$ if $pred_P(v) \in N_{start}$. If so we also check if $half_p(v)$ is true. 

If we find a vertex $v$ for which both conditions hold, we join the cycle here. To do this look at $N_{end}$ and take a vertex $q \in$  \nwn{$c_{end}, 40\log n$}  such that $half_P(q)$ is false and $pred_P(q) \nin U$. Then we add the cycle to the path by constructing the new path $P_{new} = (p_{start},..., pred_P(v)) + (c_{start},..., c_{end}) + (q,...,p_{end}) +  (v,..., pred_P(q))$. Then add $p_{end}$, $c_{start}$ and $c_{end}$ to the used vertices $U$. (If we cannot find $q$ we abort the algorithm; it will be easy to show that the probability that this happens is negligible.)

If the check fails for all $v \in N(p_{end})$ we perform a Posa rotation. To do this is we take a $v \in N(p_{end}), v \not= p_{start}$, such that $half_P(v)$ is true and such that $pred_P(v)\nin U$, and then take a $q \in$  \nwn{$pred_P(v), 40\log n$} such that both $half_p(q)$ is false and $pred_P(q)$ is unused. We then use $v$ and $q$ to construct a new path with a new endpoint, namely $P_{new} = (p_{start},..., pred_P(v))+ (q,...,p_{end}) +  (v,..., pred_P(q))$. Now we can repeat the above procedure with $P_{new}$ and the new endpoint $p_{new end} = pred_P(q)$. (If we cannot find $v$ or $q$ we abort the algorithm; again it will be easy to show that the probability that this happens is negligible.)

\begin{figure}[t]
\centering \includegraphics[width = \textwidth]{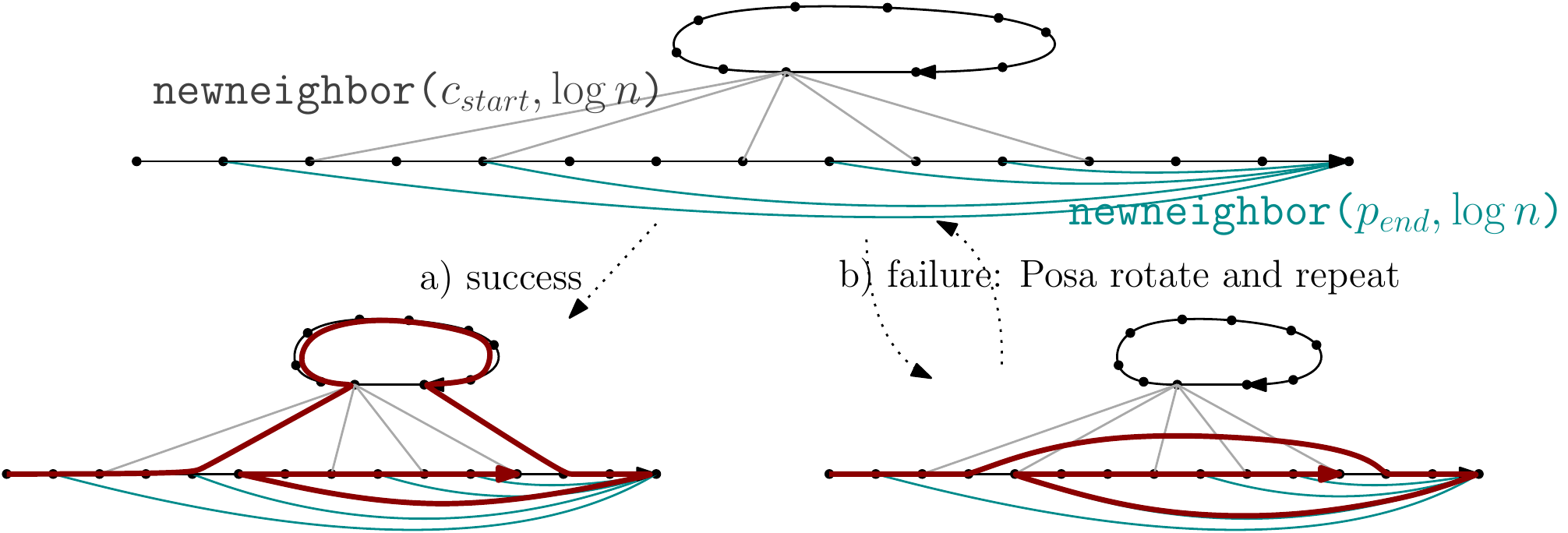} 
\caption{Incorporating a single new cycle with \texttt{AddSingleCycle}. The dark red path indicating the new Path after an iteration of the while loop.}
\label{fig:addsinglecycle}
\vspace{1em}
\end{figure} 

 \begin{algorithm}[t]
 \caption{$ JoinCycles(G, \mathcal{C} = M_1 \cup M_2 )$} \label{alg:joincycles}
 \begin{algorithmic}[1]
 \State $U\leftarrow \{\}$
 \State $C_0 \leftarrow$ first cycle of $M_1 \cup M_2$, $(c_{0, start},..., c_{0, end})$;
 \State $P \leftarrow (c_{0, start},..., c_{0, end})$;
 \State $p_{end} \leftarrow$ last vertex of P;
 \While{$|P| \le \frac{3n}{4} $}
 \State $N \leftarrow$ \nwn{$p_{end} , 40\log n$};
 \State $U \leftarrow$ add $p_{end}$;
 \State $v \leftarrow$ Search $N$ for $v$ such that $v \nin P$
 \State $(v,..., c_{i, end}) \leftarrow$ cycle of $v$;
 \State $P \leftarrow P + (v,..., c_{i, end})$;
 \State $p_{end} \leftarrow c_{i, end}$;
 \EndWhile
 \While{$|P|  \ne n $}
 \State $C_i \leftarrow$ any cycle not in $P$
 \State $AddSingleCycle(G, P, C_i, U)$ \Comment{See Algorithm~\ref{alg:addsinglecycle}}
 \EndWhile
 \State \Return
 \State // If any of the `Search' parts of the algorithm fail, we abort the algorithm and return failure.
 \end{algorithmic}
 \end{algorithm}

 \begin{algorithm}[t] 
 \caption{$ AddSingleCycle(G, P, C_i , U)$} \label{alg:addsinglecycle}
 \begin{algorithmic}[1]
 \State{// Function $half_P(v)$ returns true if and only if $v$ is in the first half of $P$;}
 \State // For any vertex $v \in P$, $pred_P(v)$ denotes the vertex before $v$ on the path $P$ ;
 \State $p_{end} \leftarrow$ last vertex of $P$;
 \State $N_{start} \leftarrow $ \nwn{$c_{start}, 40\log n$} $\cap P$;
 \State $N_{end} \leftarrow$ \nwn{$c_{end} , 40\log n$};
 \State $U \leftarrow$ add $c_{start}$ and $c_{end}$;
 \While {true}
 \State $N \leftarrow$ \nwn{$p_{end} , 40\log n$};
 \State $U \leftarrow$ add $p_{end}$;
 \If{$\exists v \in N \ \ s.t. \ \ pred_P(v) \in N_{start}$ and $half_P(v) = true$}
 \State $q \leftarrow$ Search $N_{end}$ for $q$ such that $half_p(q)= false$ and $pred_P(q) \nin U$; 
 \State $P \leftarrow  (p_{start},..., pred_P(v)) + (c_{start},..., c_{end}) + (q,...,p_{end}) +  (v,..., pred_P(q))$;
 \State \Return
 \Else 
 \State$v \leftarrow$ Search $N$ for $v$ such that $half_P(v) = true$ ;
 \State $N \leftarrow$ \nwn{$pred_P(v) , 40\log n$};
 \State $U \leftarrow$ add $pred_P(v)$;
 \State $q \leftarrow$ Search $N$ for $q$ such that $half_P(q) = false$ and $pred_P(q) \nin U$ ;
 \State $P \leftarrow  (p_{start},..., pred_P(v))+ (q,...,p_{end}) +  (v,..., pred_P(q))$; 
 \State $p_{end} \leftarrow pred_P(q)$;
 \EndIf
 \EndWhile
 \State // If any of the `Search' parts of the algorithm fail, we abort the algorithm and return failure.
 \end{algorithmic}
 \end{algorithm}

To store the path and cycles we use AVL trees with a linked list. The linked list just stores the vertices in the order as they appear in the path resp.\ cycle. For the AVL tree we take the ordering in the path/linked list as an ordering of the vertices. With this ordering at hand, the AVL tree is well defined, and it allows for searching resp. answering the query $half(v)$ in $\OO(\log n)$ time. In addition, splitting the path resp.\ concatenating two paths correspond to splitting an AVL tree at a given vertex (into a tree containing all smaller vertices and a tree containing the remaining vertices) resp. concatenate two AVL trees in which the largest vertex in one tree is smaller than the smallest vertex in the other tree. It is well known that both of these operations can be done for AVL trees in $\OO(\log n)$ time, cf. Lemma~\ref{lem:avltree} in Section~\ref{sec:datastructures} for more details. 

\begin{proposition}
	\label{prop:addcycle}
	Applying the procedure \texttt{AddSingleCycle} at most $2 \log n$ times will run in time $\OO(n)$ with high probability.
\end{proposition}

\begin{proof}
We want to bound the number of Posa rotations we need to perform while we add at most $2\log n$ cycles. Each Posa rotation occurs at the end of a while loop in the pseudocode.

To incorporate a cycle we want to find a vertex $v$ which, in the order of the path, is right after a vertex in $N_{start}$ and is in the first half of $P$. $P$ has size at least $3n/4$ so the number of vertices in the first half is at least $n/4$. A random vertex therefore has a chance of at least $1/4$ to be in the first half of $P$. So every vertex in \nwn{$c_{start}, 40\log n$} has probability at least $1/4$ independently of being in the first half of $P$ and different from the other vertices. This implies that the number of vertices in $N_{start}$ which are also in the first half of $P$ dominates a binomial distributed random variable $F \sim Bin(40\log n, 1/4)$. For $F$ we know the expectation to be $10 \log n$ and by a Chernoff bound (\ref{thm:chernoff}) the probability that $F$ is less than $ \log n$ is $\OO(n^{-2})$. We observe that where the Posa rotation happens is independent of $N_{start}$. So we apply a union bound that on fixed $\OO(n)$ many rotations of $P$ the probability that there are less than $\log n$ vertices of $N_{start}$ in the first half of $P$ is in $\OO(n^{-1})$. This implies that any call to \nwn{$p_{end}$} has a chance of at least $\log n / n$ to be right after a vertex in $N_{start}$ and also in the first half of $P$. As each call to \nwn{} is independent, the number of tries we must make is geometrically distributed with success probability $\log n / n$ and we must succeed at most $2\log n$ many times. This means the number of Posa rotations is dominated by a negative binomial distributed random variable $R \sim NB(2\log n, \log n / n)$. So by the concentration of the negative binomial distribution (Lemma~\ref{cor:nbcon}) the probability that we need to try more than $4n$ times is at most $\OO(\log^{-1} n)$. Before every Posa rotation we try \nwn{$p_{end}, 40\log n$} so $40 \log n$ tries. This proves an upper bound on the number of Posa rotations of $\OO(n / \log  n)$ with high probability.

\textbf{Posa rotation:} We summarize the operations we need to do per Posa rotation. This assumes that we already failed to find $v$ which is both after a vertex in $N_{start}$ and also in the first half of $P$. We expose $40 \log n$ new neighbors of $p_{end}$ and $40 \log n$ of the vertex before $v$ on the path, we need to Posa rotate by splitting the path twice and then joining twice. Checking whether a vertex is in $U$ and adding vertices to $U$ is a constant time operation with a lookup table. All of these operations by choice of proper datastructure (Lemma~\ref{lem:newneighbor} and \ref{lem:avltree}) are done in $\OO(\log n)$. So over all Posa rotations these sum up to a runtime of at most $\OO(n)$. Additionally we need to find the vertex $v$ in the first half of $P$ with $pred_P(v) \nin U$. Since $U$ is much smaller than $n/8$ and $|P| \ge 3n/4$ the number of possible vertices is at least $n/4$. This means that if we test a random vertex, the probability that $half_p()$ returns true and its predecessor is not in $U$ is at least $1/4$. So the number times we need to call $half_P()$ is dominated by a geometric distribution with success probability $1/4$. Similarly to find the vertex $q$ in the second half of $P$ with $pred_P(q) \nin U$, the number of times we need to call $half_P()$ is also dominated by a geometric distribution with success probability $1/4$. So over all rotations, the number of times we need to call $half_P()$ is dominated by a negative binomial distribution $H \sim NB(2 \cdot \OO(n/ \log n), 1/4)$. So by the concentration of the negative binomial distribution (Lemma~\ref{cor:nbcon}) the probability that we need to call $half_P$ more than  $\OO(n/ \log n)$ times is $\OO(\log n / n)$. And since we can perform $half_P()$ in time $\OO(\log n)$ by Lemma~\ref{lem:avltree} these have a total runtime of $\OO(n)$ with high probability.

\textbf{Incorporating cycles:} Very similarly we bound the time we need to incorporate the cycles. To find the vertex $v$ which in the order of the path is right after a vertex in $N_{start}$ and is in the first half of $P$ we need to call $half_P$ until we succeed. Note that since $|N_{start}| \le 40\log n$ and as we proved above at least $\log n$ vertices of them are in the first half of $P$, every call to $half_P()$ from a random vertex after a vertex in $N_{start}$ has a chance of succeeding of at least $1/40$. This means the number of times we call $half_P$ is again dominated by a negative binomial distribution $NB(2\log n, 1/40)$ and this runtime is negligible with high probability. As we only incorporate a cycle $2 \log n$ times, also the join and split operations as well as the exposing of $N_{end}$ and searching for $q$ are negligible compared to the $\OO(n)$ runtime. 

Note also that we only call \nwn{} of vertices we then add to $U$ and then not again during the entire algorithm so no vertex has \nwn{} called more than $40\log n$ times. At most $\OO(n/ \log n)$ many vertices are added to $U$, and $U$ is small enough so that it is always much smaller than $n / 8$.

This concludes the proof of Proposition~\ref{prop:addcycle}.
\end{proof}

\begin{lemma}
\label{lem:closecycle}
Given a Hamilton path we can transform it to a Hamilton cycle in $\OO(n)$ time.
\end{lemma}
\begin{proof}
Calling the algorithm \texttt{AddSingleCycle} with the cycle being $p_{start}$, but instead looking for $v$ such that a vertex after $v$ is in the neighborhood of $p_{start}$ instead of a predecessor gives us a cycle $C= (p_{start},..., v) + (p_{end},..., after_P(v))$. Analysis of runtime equivalent to the analysis of \texttt{AddSingleCycle}.
\end{proof}

Propositions \ref{prop:fastalgo} and \ref{prop:addcycle} as well as Lemma~\ref{lem:closecycle} show that all  components of the algorithm run in time $\OO(n)$. It is also easy to check that both phases together require at most $ 50 \log n$ calls to \nwn{} from any fixed vertex, so the assumptions of Lemma~\ref{lem:newneighbor} do hold. So choosing $C$ large enough, say $C=200$, suffices to guarantee that with high probability the random graph is such that all vertices have more neighbors than we query. This thus concludes the proof of Theorem~\ref{thm:main}.

\section{Datastructures}
\label{sec:datastructures}

In this section we give the details of the data structures that we used within our algorithm. 

\subsection{Querying a new vertex}

As explained above, we assumed throughout the analysis of our algorithm that we have access to
a function \nwn{v}, that returns for a given vertex $v$ a neighbor $w$ that is {\em uniformly} distributed in $V-v$ and whose result is {\em independent} from all previous calls.

\newneighbor

\begin{proof} To realize such a function \nwn{}, it is important to make the adjacency lists independent of each other. To realize this we transform the graph $G$ (which is distributed as a random graph $G_{n,p}$) into a directed graph $G'$ distributed as $D_{n, p/2}$ (in which each directed edge is present independently with probability $p/2$). It is well know how this can be done. In particular, we can sample $D_{n, p/2}$ from $G_{n,p}$ as a subgraph such that every edge in the directed graph is also an undirected edge in the $G_{n,p}$. 
More precisely, we do the following for every edge $\{i,j\}$ of G: with probability
\begin{align*} 
\frac{1}{2} - \frac{p}{4} \  & set \  (i,j) \in G' \ and \ (j,i) \not \in G'\\
\frac{1}{2} - \frac{p}{4} \  & set \   (i,j) \not \in G' \ and \ (j,i) \in G'\\
\frac{p}{4} \  & set \   (i,j) \in G' \ and \ (j,i) \in G'\\
 \frac{p}{4}  \  & set \  (i,j) \not \in G' \ and \ (j,i) \not \in G'
\end{align*}

In order to be consistent with the transformation from $G$ to $G'$ and to not lose too much time we only direct the edges once we see it for the first time. To recall the made decision, we store the random choices of the edges that we have we encountered so far into a hashtable. Thus, we can check for each edge that we obtain from querying the adjacency list of a vertex in $G$, whether we have seen this edge already and if so, which orientation we have chosen. The hash table has size $n$ and we use a hashfunction which is $4$-wise independent. This way the variance of the number of collisions is equal to a random function, and therefore the time we need for hashing is $\OO(n) + \OO(number\ of\ collisions) = \OO(n)$, which can be seen by applying Chebyshev's inequality. A more detailed argument of why linear probing with hash functions works in this context can be found in \cite{pagh2007linear, thorup2004tabulation}.

Finally, we want the distribution of the next edge to be uniform among the vertices. For this we need to resample from the already seen edges. Assuming we have revealed $d$ many edges from $v$ we flip a biased coin. With probability $d/(n-1)$ we retake an old neighbor and output it, one chosen uniformly at random, and with probability $1-d/(n-1)$ we take the next vertex in the adjacency list (which is also in $D_{n,p}$). Otherwise, we return one of the previously seen neighbors uniformly at random. In this way any vertex has probability exactly $1/(n-1)$ to be returned by \nwn{v}. 

\end{proof}

\subsection{Expansion}
\label{sec:expansion}

What we need from the random graph are properties of good expansion. Given the adjacency list of a vertex $v$ we define the $d$-neighborhood $N_d(v) \subseteq V(G)$ to be the set of the first $\lceil d\rceil$ calls to the function \bnwn{$v$}. In the analysis of the algorithm we make use of the following lemma.

\begin{lemma}[Neighborhood Lemma]
\label{lem:neighborhood}
Let $G_{n/2,n/2,p}$ be a random bipartite graph with $p\ge \frac{C~\log(n)}{n}$ and partitions $A$ and $B$. Then with high probability we have for all subsets $A' \subseteq A$ that the $d$-neighborhood of $A'$ is of size at least

\begin{equation} \label{eqn:rndm}
|N_{d(A')}(A')|  \ge \frac{1}{100} |A'| \cdot d(A') ,
\end{equation}
where $d(A') = min(\sqrt{\frac{n}{|A'|}} , \log(n))$. 
\end{lemma}

\begin{proof} The proof follows from a straight forward calculation of probabilities. Let us assume by contradiction there exists a set $A' \subseteq A$ with $|N_{d(A')}(A')|  < \frac{1}{100} |A'| \cdot d(A') $. Then there is a set $B' \subseteq B$ of size $|B'| = \frac{1}{100} |A'| \cdot d(A')$ containing this $d$-neighborhood, $N_{d(A')}(A') \subseteq B'$. So this is a probability we want to bound from above. The probability for a single vertex in $A'$ to have its $d$-neighborhood contained in a fixed set $B'$ is $\left( \frac{|B'|}{n}\right)^{d(A')}$ since the $d$-neighborhood is $d(A')$ vertices chosen from $B$ uniformly and independently at random. The probability for two specific sets $A' \subseteq A$ and $B' \subseteq B$ to have this property is $\left( |B'|/ n\right)^{|A'|\cdot d(A') }$. Now take the union bound over all possible sets $A'$ and $B'$ (with $|B'| = \frac{1}{100} |A'| \cdot d(A')$):
\[Pr[\text{(\ref{eqn:rndm}) false}] \le \sum_{A',B'} Pr[B' \text{ contains $ N_{d(A')}(A')$}] =    \sum_{i=1}^{n}        {\binom{n}{i}}     {\binom{n}{\frac{1}{100} i \ d(i) }}   \left( \frac{\frac{1}{100} i \cdot d(i)}{n}\right)^{ i \cdot d(i) }     \]
Then we apply an approximation for the binomial coefficients: $\binom{n}{k} \le \left( \frac{en}{k} \right)^k$. We see that $\frac{1}{4}\log \frac{100n}{i \cdot d(i)} \ge \log(e \cdot n/i) / d(i)$ so

\[Pr[\text{(\ref{eqn:rndm}) false}]  \le  \sum_{i=1}^{n} exp \left(  i \cdot  d(i)\cdot \left( - \frac{1}{2} \log\left(\frac{100  n}{i\cdot d(i)} \right) \right)  \right)\]

Now $d(i)$ is a known function of $i$. So we distinguish between two cases. When $i \ge n/\log^2(n)$, then $d(i) = \sqrt{n/i}$. And we can calculate ($i\le n$)
\[\sum_{i=1}^{n}            \left(  \frac{ i^{1/4} \cdot  n^{1/4}  }{ 10 \cdot n^{1/2} } \right)  ^{\sqrt{n \ i}}  \le \sum_{i=1}^{n}            \left( \frac{1 }{ 10 } \right)  ^{\sqrt{n \ i}}  \le \OO(n^{-2}) \]
On the other hand if $i \le n/\log^2(n)$, then $d(i) = \log(n)$. And we can calculate 
\[\sum_{i=1}^{n/\log^2(n)}            \left(\frac{ i^{1/2} \cdot  \log(n)^{1/2}  }{ 10 \cdot n^{1/2} } \right)  ^{i \log(n)} \le\sum_{i=1}^{n/\log^2(n)}            \left( \frac{1  }{ 10 \cdot \log(n)^{1/2} } \right)  ^{i \log(n)}   \le \OO(n^{-2})\]
Together this implies that the lemma holds for random graphs with probability $ \ge 1-  \OO(n^{-2})$.

\end{proof}

\subsection{AVL Trees}

\begin{lemma}[AVL Trees]
\label{lem:avltree}
We can store a path (or cycle) in an AVL tree joint with a linked list datastructure and can perform the following operations (where we view the cycle as a path split at an arbitrary point):
\begin{itemize}
\item For any vertex $v$ find the vertex preceding or succeeding it in the path in constant time $\OO(1)$
\item For any vertex $v$ searching the path it is in and determining whether it is in the first or second half of it in time $\OO(\log n)$
\item Split the path into two paths in time $\OO(\log n)$
\item Concatenate two paths into one by adding the endpoint of one to the start of the other in time $\OO(\log n)$
\end{itemize}
\end{lemma}

\begin{proof}
We combine an AVL tree, which is a balanced binary search tree, with a linked list. The AVL tree is built on the order sequence of the path as if numbering the vertices along the path from $1$ to $|P|$. The linked list ensures that going forward and backward on the path is done in constant time, where the AVL tree can perform search (for the half function) in $\OO(\log n)$. A split of the path is nothing other than splitting the AVL tree at a leaf node into two trees such that all the nodes smaller go into one tree and all the nodes larger go into the other. The concatenate is the inverse of the split and only requires attaching the smaller tree to the larger one at the appropriate node and rebalancing up to the root. Both operations run in $\OO(\log n)$ time.AVL trees are by now a part of basic datastructure lectures and in particular the split and concatenate operations can be found e.g. in the book by Knuth~\cite{knuth1998art} see page 473, which also cites from \cite{crane1972linear} or more generally on AVL trees see \cite{pfaff1998introduction}.
\end{proof}

\section{Concluding remarks}

In this paper we presented a simple randomized algorithm based on iterative random walks that construct a Hamilton cycle in time linear in the number of vertices. Our algorithm is based on first building (two) random perfect matchings. The key idea here is to expose more and more edges of the currently unmatched vertices, where the exact number of these exposed neighbors is a function of the currently unmatched vertices. Our analysis requires that the density of the random graph $G_{n,p}$ is at least $p\ge \frac{C \log n}{n}$. $C$ is chosen such that we have a sufficient minimum degree with high probability.

We leave it as an open question whether our approach can be modified to also find Hamilton cycles in $G_{n,p}$ for $p$ at the threshold for existence of Hamilton cycles. This certainly requires additional ideas.


\bibliographystyle{plainurl}
\bibliography{hamcycles}

\begin{thebibliography}{10}

\bibitem{ajtai1985first}
Mikl{\'o}s Ajtai, J{\'a}nos Koml{\'o}s, and Endre Szemer{\'e}di.
\newblock First occurrence of {H}amilton cycles in random graphs.
\newblock {\em North-Holland Mathematics Studies}, 115(C):173--178, 1985.

\bibitem{allen2015tight}
Peter Allen, Julia B{\"o}ttcher, Yoshiharu Kohayakawa, and Yury Person.
\newblock Tight hamilton cycles in random hypergraphs.
\newblock {\em Random Structures \& Algorithms}, 46(3):446--465, 2015.

\bibitem{alon2020finding}
Yahav Alon and Michael Krivelevich.
\newblock Finding a {H}amilton cycle fast on average using rotations and
  extensions.
\newblock {\em Random Structures \& Algorithms}, 57(1):32--46, 2020.

\bibitem{angluin1979fast}
Dana Angluin and Leslie~G Valiant.
\newblock Fast probabilistic algorithms for {H}amiltonian circuits and
  matchings.
\newblock {\em Journal of Computer and System Sciences}, 18(2):155--193, 1979.

\bibitem{arratia2003logarithmic}
Richard Arratia, Andrew~D Barbour, and Simon Tavar{\'e}.
\newblock {\em Logarithmic combinatorial structures: a probabilistic approach},
  volume~1.
\newblock European Mathematical Society, 2003.

\bibitem{arratia1992cycle}
Richard Arratia and Simon Tavar{\'e}.
\newblock The cycle structure of random permutations.
\newblock {\em The Annals of Probability}, pages 1567--1591, 1992.

\bibitem{bohman2009hamilton}
Tom Bohman and Alan Frieze.
\newblock {H}amilton cycles in 3-out.
\newblock {\em Random Structures \& Algorithms}, 35(4):393--417, 2009.

\bibitem{bollobas1987algorithm}
Bela Bollobas, Trevor~I. Fenner, and Alan~M. Frieze.
\newblock An algorithm for finding {H}amilton paths and cycles in random
  graphs.
\newblock {\em Combinatorica}, 7(4):327--341, 1987.

\bibitem{bollobas1990hamilton}
Bela Bollob{\'a}s, Trevor~I. Fenner, and Alan~M. Frieze.
\newblock {H}amilton cycles in random graphs with minimal degree at least k.
\newblock {\em A tribute to Paul Erdos, edited by A. Baker, B. Bollob{\'a}s and
  A. Hajnal}, pages 59--96, 1990.

\bibitem{crane1972linear}
C.A. Crane.
\newblock {\em Linear Lists and Priority Queues as Balanced Binary Trees}.
\newblock Computer Science Department. Department of Computer Science, Stanford
  University., 1972.

\bibitem{fenner1984hamiltonian}
Trevor~I Fenner and Alan~M Frieze.
\newblock {H}amiltonian cycles in random regular graphs.
\newblock {\em Journal of Combinatorial Theory, Series B}, 37(2):103--112,
  1984.

\bibitem{ferber2016finding}
Asaf Ferber, Michael Krivelevich, Benny Sudakov, and Pedro Vieira.
\newblock Finding hamilton cycles in random graphs with few queries.
\newblock {\em Random Structures \& Algorithms}, 49(4):635--668, 2016.

\bibitem{frieze2019hamilton}
Alan Frieze.
\newblock {H}amilton cycles in random graphs: a bibliography.
\newblock {\em arXiv preprint arXiv:1901.07139}, 2019.

\bibitem{frieze1996generating}
Alan Frieze, Mark Jerrum, Michael Molloy, Robert Robinson, and Nicholas
  Wormald.
\newblock Generating and counting hamilton cycles in random regular graphs.
\newblock {\em Journal of Algorithms}, 21(1):176--198, 1996.

\bibitem{frieze1988finding}
Alan~M Frieze.
\newblock Finding hamilton cycles in sparse random graphs.
\newblock {\em Journal of Combinatorial Theory, Series B}, 44(2):230--250,
  1988.

\bibitem{gurevich1987expected}
Yuri Gurevich and Saharon Shelah.
\newblock Expected computation time for {H}amiltonian path problem.
\newblock {\em SIAM Journal on Computing}, 16(3):486--502, 1987.

\bibitem{janson2011random}
Svante Janson, Tomasz Luczak, and Andrzej Rucinski.
\newblock {\em Random graphs}, volume~45.
\newblock John Wiley \& Sons, 2011.

\bibitem{knuth1998art}
Donald~E Knuth.
\newblock {\em The art of computer programming: Volume 3: Sorting and
  Searching}.
\newblock Addison-Wesley Professional, 1998.

\bibitem{komlos1983limit}
J{\'a}nos Koml{\'o}s and Endre Szemer{\'e}di.
\newblock Limit distribution for the existence of {H}amiltonian cycles in a
  random graph.
\newblock {\em Discrete mathematics}, 43(1):55--63, 1983.

\bibitem{korshunov1976solution}
Aleksei~Dmitrievich Korshunov.
\newblock Solution of a problem of erd{\H{o}}s and renyi on {H}amiltonian
  cycles in nonoriented graphs.
\newblock {\em Doklady Akademii Nauk}, 228(3):529--532, 1976.

\bibitem{maples2012number}
Kenneth Maples, Ashkan Nikeghbali, and Dirk Zeindler.
\newblock On the number of cycles in a random permutation.
\newblock {\em Electron. Commun. Probab.}, 17:13 pp., 2012.

\bibitem{montgomery2019hamiltonicity}
Richard Montgomery.
\newblock Hamiltonicity in random graphs is born resilient.
\newblock {\em Journal of Combinatorial Theory, Series B}, 139:316--341, 2019.

\bibitem{nenadov2019resilience}
Rajko Nenadov, Angelika Steger, and Milo{\v{s}} Truji{\'c}.
\newblock Resilience of perfect matchings and hamiltonicity in random graph
  processes.
\newblock {\em Random Structures \& Algorithms}, 54(4):797--819, 2019.

\bibitem{pagh2007linear}
Anna Pagh, Rasmus Pagh, and Milan Ruzic.
\newblock Linear probing with constant independence.
\newblock In {\em Proceedings of the thirty-ninth annual ACM symposium on
  Theory of computing}, pages 318--327, 2007.

\bibitem{pfaff1998introduction}
Ben Pfaff.
\newblock An introduction to binary search trees and balanced trees.
\newblock {\em Libavl Binary Search Tree Library}, 1:19--20, 1998.

\bibitem{robinson1994almost}
Robert~W. Robinson and Nicholas~C. Wormald.
\newblock Almost all regular graphs are {H}amiltonian.
\newblock {\em Random Structures \& Algorithms}, 5(2):363--374, 1994.

\bibitem{rubinfeld2011sublinear}
Ronitt Rubinfeld and Asaf Shapira.
\newblock Sublinear time algorithms.
\newblock {\em SIAM Journal on Discrete Mathematics}, 25(4):1562--1588, 2011.

\bibitem{shamir1983many}
Eli Shamir.
\newblock How many random edges make a graph {H}amiltonian?
\newblock {\em Combinatorica}, 3(1):123--131, 1983.

\bibitem{sudakov2008local}
Benny Sudakov and Van~H Vu.
\newblock Local resilience of graphs.
\newblock {\em Random Structures \& Algorithms}, 33(4):409--433, 2008.

\bibitem{thomason1989simple}
Andrew Thomason.
\newblock A simple linear expected time algorithm for finding a hamilton path.
\newblock {\em Discrete Mathematics}, 75(1-3):373--379, 1989.

\bibitem{thorup2004tabulation}
Mikkel Thorup and Yin Zhang.
\newblock Tabulation based 4-universal hashing with applications to second
  moment estimation.
\newblock In {\em SODA}, volume~4, pages 615--624, 2004.

\end{thebibliography}

\newpage


\appendix

\section{Concentration Inequalities}

We mention here some well-known inequalities used to show concentration on random variables. By the concentration of a random variable $X$ we usually mean there exist constants $c$ and $C$ such that with high probability $c \E[X] \le X \le C \E[X]$. Often also we ask $C$ to be $2$. Although these are by no means new insights, we present them here for completion and as a help for the reader. The Chernoff and Chebyshev inequality can be found e.g. in the book \cite{janson2011random}.

\begin{theorem}[Chernoff Inequality] 
\label{thm:chernoff}
If $X$ is distributed as a binomial random variable $X \sim Bin(n,p)$ and $0 < \eps \le 3/2$, then
\[Pr[ |X-\E[X]|  \ge \eps \E[X] ] \le 2 e^{-\frac{\eps^2 \E[X]}{3}} \]
\end{theorem}

\begin{theorem}[Chebyshev Inequality] 
\label{thm:chebyshev}
For any random variable $X$ for which the variance $Var[X]$ exists, 

\[Pr[ |X-\E[X]|  \ge t ] \le \frac{Var[X]}{t^2} \]

\end{theorem}

We use the term negative binomial distribution in the analysis and since this is defined slightly differently sometimes we give here the definition we use.

\begin{definition}
Let $X_i$ be independent bernoulli random variables with probability of being one is $p$ for any $i \in \NN$. For any $r \in \NN$ let $Y$ be index of the $r$-th $X_i$ which evaluates to $1$. Then $Y$ has a negative binomial distribution $NB(r, p)$.
\end{definition}

We observe that a negative binomial distribution $Y \sim NB(r, p)$ is equivalent to $Y$ being distributed as the sum of $r$ geometric random variables with success probability $p$. Further a simple corollary from the Chebyshev inequality:

\begin{corollary}
\label{cor:nbcon}
For a negative binomial distributed variable $Y \sim NB(r, p)$

\[Pr\left[ Y \ge 2 \frac{r}{p} \right] \le \frac{1}{r} \]

\end{corollary}
\begin{proof}
We calculate $\E[Y] = r/p$ and $Var[Y] = r (1-p)/p^2$ and apply Chebyshev.
\[  Pr\left[ Y \ge 2 \frac{r}{p} \right] \le  Pr\left[ |Y -\E[Y]|  \ge \frac{r}{p} \right]  \stackrel{Chebyshev}{\le} \frac{1-p}{r}  \le \frac{1}{r}  \]

\end{proof}

\end{document}